\documentclass[11pt]{article}

\usepackage[margin=1in]{geometry}
\usepackage{times}
\usepackage{verbatim}
\usepackage{graphicx}
\usepackage{paralist}
\usepackage{amsmath,amsthm,amssymb,amsfonts}
\usepackage{caption}
\usepackage{setspace}
\usepackage[colorlinks,linkcolor=black,bookmarksopen,
  bookmarksnumbered,citecolor=black,urlcolor=blue]{hyperref}

\newcommand{\R}{ {\mathbb R} }

\usepackage{color}

\usepackage[normalem]{ulem}

\definecolor{darkgrn}{rgb}{0, 0.8, 0}

\newcommand{\rank}{\operatorname{rank}}
\newcommand{\Pred}{\operatorname{Pred}}

\usepackage[ruled,vlined]{algorithm2e}

\theoremstyle{plain}
\newtheorem{theorem}{Theorem}[section]
\newtheorem{lemma}[theorem]{Lemma}

\theoremstyle{definition}
\newtheorem{definition}[theorem]{Definition}
\theoremstyle{remark}
\newtheorem{remark}[theorem]{Remark}
\theoremstyle{claim}

\usepackage{gensymb}
\usepackage{mathrsfs}

\newcommand{\maxip}{\textsc{Max-IP}}
\newcommand{\maxipsp}{\textsc{Max-IP} }
\newcommand{\maxipd}{\textsc{Max-IPD}}
\newcommand{\maxipdsp}{\textsc{Max-IPD} }
\newcommand{\ip}{\textsc{IP}}
\newcommand{\ipsp}{\textsc{IP} }

\newcommand{\kip}{\textsc{$k$-IP}}
\newcommand{\kipsp}{\textsc{$k$-IP} }
\newcommand{\kipd}{\textsc{$k$-IPD}}
\newcommand{\kipdsp}{\textsc{$k$-IPD} }

\newcommand{\tipd}{\textsc{$3$-IPD}}
\newcommand{\tipdsp}{\textsc{$3$-IPD} }
\newcommand{\cZ}{\mathcal{Z}} 
\newcommand{\cP}{\mathcal{P}} 
\newcommand{\cC}{\mathscr{C}} 
\newcommand{\Win}{W_{\rm in}}
\newcommand{\Wout}{W_{\rm out}}

\newcommand{\Sig}{\operatorname{Sig}}
\newcommand{\IScr}{{\cal I}}

\usepackage{subcaption}

\title{Interesting Paths in the Mapper}
\author{
  Ananth Kalyanaraman\footnote{School of Electrical Engineering and Computer Science, Washington State University, Pullman, WA, 99164, USA; \href{mailto:ananth@wsu.edu}{ananth@wsu.edu}},
  Methun Kamruzzaman\footnote{School of Electrical Engineering and Computer Science, Washington State University, Pullman, WA, 99164, USA; \href{mailto:md.kamruzzaman@wsu.edu}{md.kamruzzaman@wsu.edu}},  and
Bala Krishnamoorthy\footnote{Department of Mathematics and Statistics, Washington State University, Vancouver, WA, 98686, USA; \href{mailto:kbala@wsu.edu}{kbala@wsu.edu}}
}
\date{\relax}

\begin{document}
\maketitle

\begin{abstract}

  Given a high dimensional point cloud of data with functions defined on the points, the Mapper produces a compact summary in the form of a simplicial complex connecting the points.
  This summary offers insightful data visualizations, which have been employed in applications to identify subsets of points, i.e., subpopulations, with interesting properties.
  These subpopulations typically appear as long paths, flares (i.e., branching paths), or loops in the Mapper.

  We study the problem of quantifying the interestingness of subpopulations in a given Mapper.
  First, we create a weighted directed graph $G=(V,E)$ using the $1$-skeleton of the Mapper.
  We use the average values at the vertices (i.e., clusters) of the target function (i.e., a dependent variable) to direct the edges from low to high values.
  We set the difference between the average values at the vertices (high$-$low) as the weight of the edge.
  Covariation of the remaining $h$ functions (i.e., independent variables) is captured by a $h$-bit binary signature assigned to the edge.
  An \emph{interesting path} in $G$ is a directed path whose edges all have the same signature.
  Further, we define the interestingness score of such a path as a sum of its edge weights multiplied by a nonlinear function of their corresponding ranks, i.e., the depths of the edges along the path.
  The goal is to value more the contribution from an edge deep in the path than that from a similar edge which appears at the start.

  Second, we study three optimization problems on this graph $G$ to quantify interesting subpopulations.
  In the problem \maxip, the goal is to find the most interesting path in $G$, i.e., an interesting path with the maximum interestingness score.
  We show that \maxipsp is NP-complete.
  For the case where $G$ is a directed acyclic graph (DAG), which could be a typical setting in many applications, we show that \maxipsp can be solved in polynomial time---in $O(mnd_{\rm in})$ time and $O(mn)$ space, where $m,n$, and $d_{\rm in}$ are the numbers of edges, vertices, and the maximum indegree of a vertex in $G$, respectively.

  In the more general problem \ip, the goal is to find a collection of interesting paths that are edge-disjoint, and the sum of interestingness scores of all paths is maximized.
  We also study a variant of \ipsp termed \kip, where the goal is to identify a collection of edge-disjoint interesting paths each with $k$ edges, and the total interestingness score of all paths is maximized.
  While \kipsp can be solved in polynomial time for $k \leq 2$, we show \kipsp is NP-complete for $k \geq 3$ even when $G$ is a DAG.
  We develop heuristics for \ipsp and \kipsp on DAGs, which use the algorithm for \maxipsp on DAGs as a subroutine, and run in $O(mnd_{\rm in})$ and $O(mkd_{\rm in})$ time for \ipsp and \kipsp, respectively.
  
 \end{abstract}

\section{Introduction} \label{sec:intro}

Data sets from many applications come in the form of point clouds often in high dimensions along with multiple functions defined on these points.
Topological data analysis (TDA) has emerged in the past two decades as a new field whose goal is to summarize such complex data sets, and facilitate the understanding of their underlying topological and geometric structure.
In this paper, we focus on \emph{Mapper} \cite{SiMeCa2007}, a TDA method that has gained significant traction across various application domains \cite{Al2012,Hietal2015,KaKaKrSc2017,Lietal2015,Lumetal2013,NiLeCa2011,Nietal2015,Ruetal2015,Saetal2014,ToOlThRaCuSc2016}.

Starting from a point cloud $X$ (typically sampled from a metric space), the Mapper studies the topology of the sublevel sets of a \emph{filter function} $f : X \to Z \subset \R$.
Starting with a cover $\cZ$ of $Z$, the Mapper obtains a cover of the domain $X$ by pulling back $\cZ$ through $f$.
This pullback cover is then refined into a connected cover by splitting each of its elements into various clusters using a clustering algorithm that employs another function $g$ defined on $X$.
A compact representation of the data set, also termed Mapper, is obtained by taking the nerve of this connected cover---this is a simplicial complex with one vertex per each cluster, one edge per pair of intersecting clusters, and one $p$-simplex per non-empty $(p+1)$-fold intersection in general.
The method can naturally consider multiple filter functions $f_i : X \to Z_i$, with covers $\cZ_i$ jointly pulled back to obtain the cover of $X$.
Equivalently, one could consider them together as a \emph{single} vector-valued filter function $\mathbf{f} : X \to Z \subset \R^h$ for $h \geq 2$.

The Mapper has been used in a growing number of applications from diverse domains recently, ranging from medicine \cite{Hietal2015,Lietal2015,NiLeCa2011,Nietal2015,Ruetal2015,Saetal2014,ToOlThRaCuSc2016} to agricultural biotechnology \cite{KaKaKrSc2017} to basketball player profiles \cite{Al2012} to voting patterns \cite{Lumetal2013}.
It is also the main engine in the data analytics software platform of the firm Ayasdi.
The key to all these success stories is the ability of Mapper to identify subsets of $X$, i.e., \emph{subpopulations}, that behave distinctly from the rest of the points.
In fact, this feature of Mapper distinguishes it from many traditional data analysis techniques based on, e.g., machine learning, where the goal is usually to identify patterns valid  for the entire data set.

Several researchers have recently studied mathematical and foundational aspects of Mapper (see Section~\ref{ssec:related}).
These results have enabled robust ways to build \emph{one} representative Mapper for any given data set.
One could follow the work of Carri\`ere et al.~\cite{CaMiOu2017}, or identify parameters corresponding to a stable range in the persistence diagram of the multiscale mapper \cite{DeMeWa2016}.
While a collection of mappers built at multiple scales, e.g., the multiscale mapper \cite{DeMeWa2016}, could provide a more detailed representation of the data, efficient summarization of the entire representation as well as the extraction of insights relevant to the application still remain challenging.
Hence working with a single Mapper could be considered the desirable setting from the point of view of most applications.

At the same time, many applications demand more precise quantification of the interesting features in the selected Mapper, as well as to track the corresponding subpopulations as they evolve along such features.
In plant phenomics \cite{KaKaKrSc2017}, for instance, we are interested in identifying specific varieties (i.e., genotypes) of a crop that show resilient growth rates as several environmental factors vary in specific ways during the entire growing season.
Each such subpopulation with the associated variation would suggest a testable hypothesis for the practitioner.
For this purpose, it is also desirable to rank these subpopulations in terms of their ``interestingness'' to the practitioner.

\vspace*{-0.07in}
\subsection{Our contributions} \label{ssec:ourcontr}

We propose a framework for quantifying the interestingness of subpopulations in a given Mapper.
For the input point cloud $X$, we assume the Mapper is constructed with $h$ filter functions $f_i : X \to Z_i$ that represent independent variables, and a target function (i.e., a dependent variable) $g : X \to Z$.
The Mapper could be a high-dimensional simplicial complex depending on the choice of covers $\cZ_i$ for $Z_i$.

\medskip
\noindent {\bf Formulation:} We create a weighted directed graph $G=(V,E)$ using the $1$-skeleton of the Mapper.
We use the average values of $g$ at the vertices (i.e., clusters) to direct the edges from low to high values.
We set the difference between the average values at the vertices (high$-$low) as the weight of the edge.
Covariation of the $h$ functions $f_i$ is captured by a $h$-bit binary signature assigned to the edge.
We define an \emph{interesting path} in $G$ as a directed path whose edges all have the same signature (all references to a ``path'' in this paper imply a simple path, i.e., no vertices are  repeated).
Further, we define the interestingness score of such a path as a sum of its edge weights multiplied by a nonlinear function of their corresponding ranks, i.e., the depths of the edges along the path.
The goal is to value more the contribution from an edge deep in the path than that from a similar edge which appears at the start.

\noindent {\bf Theoretical results:} We study three optimization problems on this graph $G$ to quantify interesting subpopulations.
In the problem \maxip, the goal is to find the most interesting path in $G$, i.e., an interesting path with the maximum interestingness score.
We show that \maxipsp is NP-complete.
For the special case where $G$ is a directed acyclic graph (DAG), which could be a typical setting in applications, we show that \maxipsp can be solved in polynomial time---in $O(m\delta_{\rm max}d_{\rm in})$ time and $O(mn)$ space where $m$ and $n$ are the number of edges and vertices respectively, and $\delta_{\rm max}$ and $d_{\rm in}$ are the diameter of $G$ and the maximum indegree of any vertex in $G$ respectively.
Note that $d_{\rm in}<n$ and $\delta_{\rm max}<n$ (as $G$ is a DAG).

In the more general problem \ip, the goal is to find a collection of interesting paths such that these paths form an exact cover of $E$ and the overall sum of interestingness scores of all paths is maximum.
The collection of paths identified by \ipsp could include some short ones in terms of number of edges.
Hence we study also a variant of \ipsp termed \kip, where the goal is to identify a collection of interesting paths each with $k$ edges for a given number $k$, an edge in $E$ is part of at most one such path, and the total interestingness score of all paths is maximum.
While \kipsp can be solved in polynomial time for $k \leq 2$, we show \kipsp is NP-complete for $k \geq 3$.
Further, we show that \kipsp remains NP-complete for $k \geq 3$ even for the case when $G$ is a DAG.
Finally, we develop heuristics for \ipsp and \kipsp on DAGs, which use the algorithm for \maxipsp on DAGs as a subroutine, and run in $O(mnd_{\rm in})$ and $O(mkd_{\rm in})$ time for \ipsp and \kip, respectively.

\medskip
\noindent {\bf Software and application:}
In this paper, we focus primarily on the theoretical aspects of interesting paths.
However, we are also actively developing and maintaining an open source software repository that integrates all our ongoing implementations, and their experimental evaluations and applications. 
Section~\ref{sec:Implementation} briefly reports on this ongoing effort, including our use of plant phenomics \cite{houle2010phenomics} as a novel application domain for test, validation, and discovery.

\subsection{Related work} \label{ssec:related}

\vspace*{-0.04in}
In most previous applications of Mapper \cite{Al2012,Hietal2015,Lietal2015,Lumetal2013,NiLeCa2011,Nietal2015,Ruetal2015,Saetal2014,ToOlThRaCuSc2016}, interesting subpopulations are characterized by features (paths, flares, loops) identified in a visual manner.
As far as we are aware, our work proposes the first approach to rigorously quantify the interesting features, and to rank them in terms of their interestingness.
The works of Carri\`ere et al.~\cite{CaMiOu2017,CaOu2017} present a rigorous theoretical framework for 1-dimensional Mapper, where the features are identified as points in an extended persistence diagram.
But this line of work does not address the relative importance of the features in the context of the application generating the data.
Our work can be considered as a post processing of the Mapper identified by the methods of Carri\`ere et al.
While our framework can naturally consider multiple filter functions for a given Mapper, we do not address the stability of the interesting paths identified.

The interesting paths problems we study are related to the class of nonlinear shortest path and minimum cost flow problems previously investigated.
Non-additive shortest paths have been studied \cite{TsZa2004}, and the more general minimum concave cost network flow problem has been shown to be NP-complete \cite{GuPa1991,TuGhMiVa1995}.
At the same time, versions of shortest path or minimum cost flow problems where the contribution of an edge depends nonlinearly on its position or depth in the path appear to have not received much attention.
Hence the specific problems we study should be of independent interest as a new class of nonlinear longest path (equivalently, shortest path) problems.

\section{Methods} \label{sec:methods}

We refer the reader to the original paper by Singh et al.~\cite{SiMeCa2007} for background on Mapper, and recent other work \cite{CaMiOu2017,CaOu2017,DeMeWa2016} for related constructions.
For our purposes, we start with a single Mapper $M$ that is a possibly high dimensional simplicial complex constructed from a point cloud $X$ using filter functions $f_i : X \to Z_i \subset \R$ for $i=1,\dots,h$ and another function $g : X \to \R$. 
In the setting of a typical application, $g$ could represent a dependent variable whose relationship with the independent variables represented by $f_i$ is of interest.
We assume $g$ is used for clustering within the mapper framework.

\subsection{Interesting paths and interestingness scores} 
\label{sec:IS}
Each vertex in $M$ represents a cluster of points from $X$ that have similar values of function $g$, the dependent variable.
An edge in $M$ connects two such clusters containing a non-empty intersection of points.
By definition, each edge in $M$ connects clusters belonging to distinct elements of the pullback cover of $X$, and hence the corresponding values of the filter functions $f_i$ also change when moving along the edge.
Therefore, by following a trail of vertices (i.e., clusters) whose average $g$ values are monotonically varying, we can capture subpopulations that gradually or abruptly alter their behavior as measured by $g$ under continuously changing filter intervals.
For instance, a plant scientist interested in crop resilience may seek to identify a subset of crop individuals/varieties that exhibit sustained or accelerated growth rates ($g$) despite potentially adverse fluctuations in temperature ($f_1$) and humidity ($f_2$). 
We formulate the problem of identifying such subpopulations as that of finding interesting edge-disjoint paths in a directed graph.

\subsubsection{Graph Formulation} \label{sec:GraphForm}

We construct a weighted directed graph $G=(V,E)$ representation of the $1$-skeleton of $M$ along with some additional information.
Let $n=|V|$ and $m=|E|$ denote the numbers of vertices and edges in $G$, respectively.
We set $V$ as the set of vertices ($0$-simplices) of $M$, and $E$ as the set of edges ($1$-simplices) of $M$.
We assign directions and weights to the edges as follows.
Each vertex $u \in V$ denotes a subset of points from $X$ that constitute a partial cluster.
We denote this subset as $X(u)$. 
We let $g(u)$ and $f_i(u)$ denote the average values of the clustering function $g$ (dependent variable) and the filter function $f_i$, respectively, for all points in $u$:
\[ g(u)=\frac{\Sigma_{x \in X(u)} \, {g(x)}}{|X(u)|} \, ~\mbox{ and }~
   f_i(u) = \frac{\Sigma_{x \in X(u)} \, {f_i(x)}}{|X(u)|}\,, ~i=1,\dots,h.\]
We let $\omega(u)$ represent the weight of a vertex $u$. 
In this paper, we set $\omega(u)$ to be equal to $g(u)$.
For an edge $e=(u,v)$ in $E$, we assign as its weight as
$\omega(e) = |\omega(u)-\omega(v)| = |g(u)-g(v)|$.
Notice $\omega(e) \geq 0$ for all edges $e$ in $G$. 

Each edge is also associated with a direction, determined by one of the two rules illustrated in Figure~\ref{fig:path_dir_rule}. 
The simpler rule, termed {\bfseries Rule a}, directs the edge from the vertex with lower weight to the vertex with higher weight
(``a'' for \emph{ascending} weight).
If the vertex weights are equal, one of the two directions is chosen arbitrarily.
The other rule, termed {\bfseries Rule b}, handles differently the case where the weights of the two vertices are close, as defined by a cutoff $\tau \geq 0$.
In this case, the edge is allowed to be bidirectional (``b'' for \emph{bidirectional}), i.e., both the forward and backward edges are added with the same weight.
If the weights of the vertices are not close, {\bfseries Rule a} is followed.
The latter relaxed scheme is motivated toward capturing more robust interesting paths in practice. 

Note that if only {\bfseries Rule a} is followed, the resulting graph is guaranteed to be a directed acyclic graph (DAG). 
On the other hand, if the relaxed scheme 
{\bfseries Rule b} is followed, then the resulting graph is directed, 
i.e., it could potentially be cyclic.
Consequently, the optimization problems we address in this paper (see Section~\ref{sec:OptProblems}) will be presented for both DAG and directed graph inputs.

\medskip
\begin{figure*}[ht!]
    \centering
    \begin{subfigure}{0.475\textwidth}
        \centering
        \vspace*{0.15in}
        \includegraphics[keepaspectratio=yes, width=.9\linewidth]{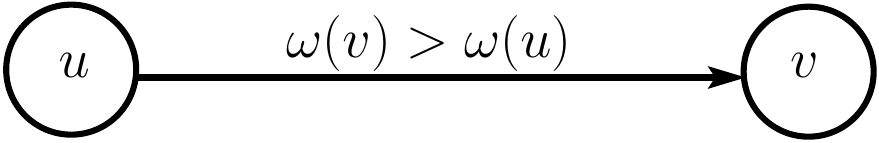}
        \vspace*{0.17in}
        \caption*{{\bfseries Rule a:} Direct edge from $u$ to $v$ if $\omega(v)>\omega(u)$,  and from $v$ to $u$ if $\omega(v)<\omega(u)$. If the weights are equal, then one of the two directions is chosen arbitrarily.}
        \label{fig:dir_path}
    \end{subfigure}%
    \hfill 
    \begin{subfigure}{0.475\textwidth}
        \centering
        \includegraphics[keepaspectratio=yes, width=.9\linewidth]{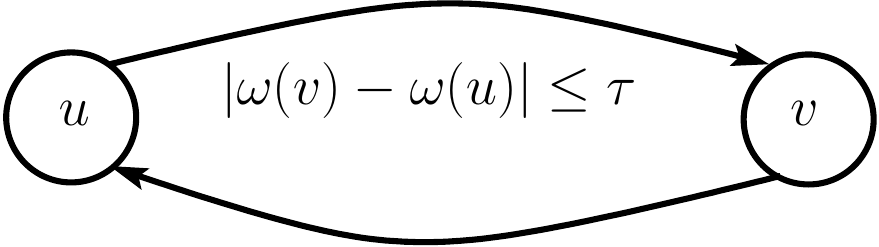}
        \caption*{{\bfseries Rule b:} Two opposite directed edges between nodes $u$ and $v$ if $|\omega(v)-\omega(u)|\leq\tau$, and directed as per Rule a otherwise.}
         \label{fig:bidir_path}
    \end{subfigure}
     \caption{
	Two rules for directing an edge, based on vertex weights and a user-defined cutoff ($\tau$).  
}   
    \label{fig:path_dir_rule}
\end{figure*}

We assign a $h$-bit binary signature $\Sig(e)=b_1b_2 \dots b_h$ to the oriented edge $e=(u,v)$ (i.e., $e: u \to v$) to capture the covariation of $g$ and the filter functions $f_i$.
We set $b_i=1$ if $f_i(u) \leq f_i(v)$, and $b_i=0$ otherwise.
If the edge is bidirected, then the signature is used as a ``wildcard''---the signature is not predetermined, and is chosen to match any candidate signature as determined by the interesting path detection algorithm.

\begin{definition} \label{def:IPath}
An \emph{interesting $k$-path} for a given $k$ with $1 \leq k \leq n-1$ is a directed path $P=[e_{i_1}, \dots, e_{i_k}]$ of $k$ edges in $G$, such that $\Sig(e_r)$ is identical for all $r=i_1,\dots,i_k$.
  An \emph{interesting path} is a path of arbitrary length in the interval $[1,n-1]$.
\end{definition}  

Note that in the above definition, a directed path is one where all edges in the path have the same direction.
Hence we have the flexibility to use a bidirected edge as part of a directed path either in the forward or reverse direction, but not both.

\begin{definition} \label{def:IScr}
  Given an interesting $k$-path $P=[e_{i_1}, \dots, e_{i_k}]$ in $G$ as specified in Definition \ref{def:IPath}, we define its \emph{interestingness score} as follows.
  \begin{equation} \label{eq:IScr}
    \IScr(P) = \sum_{r=1}^k \omega(e_{i_r}) \times \log(1+r). \, 
  \end{equation}
  In particular, the contribution of an edge $e \in P$ to $\IScr(P)$ is set to $\omega(e) \times \log(1+\rank(e,P))$, where $\rank(e,P)$ is the rank or order of edge $e$ as it appears in $P$.
\end{definition}

Intuitively, we use the rank of an edge as an inflation factor for its weight---the later an edge appears in the path, the more its weight will count toward the interestingness of the path.
This logic incentivizes the growth of long paths.
The log function, on the other hand, helps temper this growth in terms of number of edges.
  Inflation of weights for edges that appear later in the path 
  is motivated by the potential interpretation of interesting paths in the context of real world applications.
  For instance, while analyzing plant phenomics data sets \cite{KaKaKrSc2017,martin2013plant}, we expect plants to show accelerated growth spurts later in the growth season. 
  Plants showing such spurts later in the season are potentially more interesting to the practitioner than ones showing a steady growth rate throughout the season.

\begin{remark}
 The above framework can be modified easily to characterize robust interesting paths, where the signature matching condition is relaxed---for instance, $b_i=1$ if $0.9 f_i(v) \leq f_i(u) \leq f_i(v)$.
\end{remark}

\begin{remark}
  While we assume $g$ and $f_i$ are functions from $X$ to $\R$, our framework could handle more general functions as well. 
  If some $f_i$ is a vector-valued function, for instance, we could first compute pairwise distances of the points in $X$ using $f_i$, and then assign to each point in $X$ its average distance to all other points as a ``surrogate'' function.
\end{remark}

\begin{remark} \label{rem:gclust}
  The framework applies without change to cases where $g$ is used along with other functions to cluster.
  In fact, $g$ could be used also as a filter function 
  as long as it is used for clustering.
\end{remark}

\subsubsection{Optimization Problems}
\label{sec:OptProblems}

We now present multiple optimization problems with the broader goal of identifying interesting path(s) that maximize interestingness score(s).

\smallskip
\hspace*{-0.15in}%
\begin{tabular}{rll}
  {\bfseries \maxip}: &\parbox[t]{5.5in}{Find an interesting path $P$ in $G$ such that $\IScr(P)$ is maximized.} \\
\vspace*{-0.07in}  \\
{\bfseries \kip}: & \parbox[t]{5.5in}{For a given $k$ between $1$ and $n-1$, find a collection $\cP$ of interesting $k$-paths such that each $e \in E$ is part of at most one $P \in \cP$, and the total interestingness score $\IScr(\cP) = \sum_{P \in \cP} \IScr(P)$ is maximized.}
\vspace*{0.05in}  \\
{\bfseries \ip}: &\parbox[t]{5.5in}{Find a collection $\cP$ of interesting paths in $G$ such that the total interestingness score $\IScr(\cP) = \sum_{P \in \cP} \IScr(P)$ is maximized ($\cP$ will exactly cover $E$, i.e., each $e \in E$ is part of exactly one $P \in \cP$).} \\
\end{tabular}

\medskip
  Both \ipsp and \kipsp produce edge-disjoint collections of interesting paths.
  In \ip, every edge in $G$ is part of an interesting path in $\cP$.
  But this setting might include several short (in number of edges) interesting paths.
  In \kip, each interesting path found has exactly $k$ edges, and some edges in $E$ might not be part of any interesting $k$-path in $\cP$.
  Hence the paths identified by \kipsp are likely to be more meaningful in practice. 

\begin{remark}
  The $\log(1+\rank)$ factor in the interestingness score in Equation (\ref{eq:IScr}) makes each of the above optimization problems nonlinear.
  At the same time, the type of nonlinearity introduced here is distinct from the ones studied in the literature, e.g., in non-additive shortest paths \cite{TsZa2004}, or in minimum concave cost flow \cite{GuPa1991,TuGhMiVa1995}.
  Hence these problems form a new class of nonlinear longest (equivalently, shortest) path problems, which would be of interest independent of their application in the context of the Mapper and TDA.
\end{remark}

\begin{remark}
  Our directed graph formulation (in Section~\ref{sec:GraphForm}) with signatures determined by $h$ filter functions could also be considered equivalently as the computation of the coboundary of the filter functions seen as a $0$-cochain with coefficients in $\R^h$ \cite{Munkres1984}.
  Further, under appropriate assumptions on the filter functions being smooth, candidates for interesting paths could be seen as flows in a gradient field \cite[\S 14]{Tu2011}.
  Under these assumptions, one could argue that the graph constructed will necessarily be a DAG, and results from Morse theory \cite{Ma2002,Mi1963} would also apply.
  But we are not assuming the functions involved are necessarily smooth.
  More importantly, the relaxed Rule b (illustrated in Figure~\ref{fig:path_dir_rule}) creates bidirectional edges, which are more appropriate in the context of real world applications.
  Further, the $\log(1 + \rank)$ factor used in defining our interestingness score (in Equation~\ref{eq:IScr}) is not captured by default approaches for maximal flows in gradient fields or by Morse theory.
\end{remark}

\section{The \maxipsp Problem} \label{sec:MaxIP}

The goal of \maxipsp is to identify an interesting path with the maximum interestingness score. 
We show \maxipsp is NP-complete on directed graphs, but is in P on directed acyclic graphs (DAGs).

\subsection{\maxipsp on directed graphs} \label{sec:DirectedMaxIP}

In the decision version of \maxipsp termed \maxipd, we are given a directed graph $G=(V,E)$ with edge weights $\omega(e) \geq 0$ and signatures $\Sig(e)$ for $e \in E$ and a target score $s_0 \geq 0$.
The goal is to determine if there exists an interesting path $P$ in $G$ whose interestingness score $\IScr(P) \geq s_0$.

\begin{lemma}\label{lem:MaxIPDNPC}
  \maxipdsp on directed graph $G=(V,E)$ is NP-complete. 
\end{lemma}

\begin{proof} 

  Given a path $P$ in $G$, we can verify that it is a directed simple path with all the edges of the same signature, compute its interestingness score $\IScr(P)$ using Equation (\ref{eq:IScr}), and compare it with $s_0$---all in polynomial time.
  Hence \maxipdsp is in NP.

We reduce the problem of checking if a directed graph has a directed Hamiltonian cycle (\textsc{DirHC}) to \maxipd.
\textsc{DirHC} is one of the 21 NP-complete problems originally introduced by Karp \cite{Ka1972}.
Given an instance $G=(V,E)$ of \textsc{DirHC} with $|V|=n$, we construct an instance of \maxipdsp on a directed graph $G'=(V',E')$ as follows.
We replace an arbitrary vertex $v \in V$ by two vertices $v'$ and $v''$, i.e., $V' = (V\setminus\{v\}) \cup \{v',v''\}$ and $|V'| = n+1$.
Each $(v,w) \in E$ is replaced by $(v',w)$ in $E'$ and each $(u,v) \in E$ is replaced by $(u,v'')$ in $E'$.
All other edges in $E$ are included in $E'$ without changes.
All edges in $E'$ are assigned unit weights and identical signatures, and we set $s_0 = \log((n+1)!)$.

We claim that $G$ has a directed Hamiltonian cycle $C$ if and only if there exists an interesting path $P'$ in $G'$ with interestingness score $\IScr(P') = s_0$.
Let $G$ have a directed Hamiltonian cycle $C$.
Then $C$ must have $n$ edges by definition, and $C$ visits (i.e., enters and leaves) each vertex in $V$ exactly once.
Hence there must exist edges $(v,w)$ and $(u,v)$ in $C$.
We construct the interesting path $P'$ in $G'$ using $(v',w), (u,v'')$, and the remaining $(n-2)$ edges in $C$.
Thus $P'$ is a directed path in $G'$ with $n$ edges.
Further, since all edges in $E'$ have unit weights and identical signatures, it is clear
from Equation (\ref{eq:IScr}) that $P'$ is indeed an interesting path in $G'$ with $\IScr(P') = s_0$.

Conversely, let $P'$ be an interesting path in $G'$ with $\IScr(P') = s_0$
(notice that $\IScr(P') > s_0$ is not possible, as nodes are not allowed to be repeated in $P'$). 
Since $P$ is an interesting path, it visits (i.e., enters and/or leaves) any vertex in $V'$ at most once.
Since all edges in $E'$ have unit weights and identical signatures, and by the definition of
interestingness score in Equation (\ref{eq:IScr}), it is clear that $P'$ must have $n$ edges.
Hence $P'$ must start with an edge $(v',w)$ and end with an edge $(u,v'')$.
Then the directed cycle $C$ in $G$ defined by the edges $(v,w), (u,v)$, and the remaining $(n-2)$ edges in $P'$ is Hamiltonian.
Hence \maxipdsp on directed graphs is NP-complete.
\end{proof}


\subsection{\maxipsp on directed acyclic graphs}  \label{ssec:MaxIPDAG}

\begin{lemma}\label{lem:MAXIP_DAG_in_P}
 \maxipsp on a directed acyclic graph $G=(V,E)$ is in P. 
\end{lemma}

\begin{proof}
  We present a polynomial time algorithm for \maxipsp on a DAG (as proof of Lemma~\ref{lem:MAXIP_DAG_in_P}).
  The input is a DAG $G=(V,E)$ with $n$ vertices and $m$ edges, with edge weights $\omega(e)\geq 0$ and signatures $\Sig(e)$ for all $e\in E$.
  The output is an interesting path $P^*$ which has the maximum interestingness score in $G$. 
  We use dynamic programming, with the forward phase computing $\IScr(P^*)$ and the backtracking procedure reconstructing a corresponding $P^*$.

  Let $T(i,j)$ denote the score of a maximum interesting path of length $j$ edges ending at edge $e_i$ for $i\in [1,m]$. 
  Since an interesting path could be of length at most $(n-1)$, we have $j\in [1,n-1]$.
  Therefore the values in the recurrence can be maintained in a 2-dimensional table of size
  $m \times (n-1)$, as illustrated in Figure \ref{fig:maxipalgoTbl}. 
The algorithm has three steps:
\begin{itemize}
\item{\bf Initialization:}
$ T(i,1) = \omega(e_i) \times \log(2)~, \mbox{where } 1\leq i\leq m$.

\item{\bf Recurrence:}
For an edge $e=(u,v) \in E$, we define a \emph{predecessor edge} of $e$ as any edge
$e^\prime\in E$ of the form $e^\prime=(w,u)$ and $\Sig(e^\prime)=\Sig(e)$. 
Let $\Pred(e)$ denote the set of all predecessor edges of $e$.
Note that $\Pred(e)$ can be possibly empty.
We define the recurrence for $T(i,j)$ as follows.
\begin{equation} \label{eq:rec}
  T(i,j) = \max_{e_{i^\prime} \in \,\Pred(e_i)}\big\{T(i^\prime,j-1) + \omega(e_i) \times \log(1+j) \big\}
\end{equation}

\item{\bf Output:}
  We report the score that is maximum in the entire table. 
  A corresponding optimal path $P^*$ can be obtained by backtracking from that cell to the  first column.


\end{itemize}

\paragraph{Proof of correctness:}
Any interesting path in $G$ can be at most $n-1$ edges long.
As a particular edge could appear anywhere along such a path, its rank can range between $1$
and $n-1$.
Hence the $m \times (n-1)$ recurrence table $T$ sufficiently captures all possibilities for each edge in $E$.
The following key observation completes the proof. 
Let $P^*(i,j)$ denote an optimal scoring path, if one exists, of length $j\in [1,n-1]$ ending at edge $e_i\in E$.
If $P^*(i,j)$ exists and if $j>1$, then there should also exist $P^*(i^\prime,j-1)$ where $i^\prime\in \Pred(e_i)$.
Furthermore, the edge $e_i$ \emph{could not} have appeared in $P^*(i^\prime,j-1)$ because $G$ is acyclic. 
Therefore, due to the edge-disjoint nature of $P^*(i^\prime,j-1)$ and the remainder of $P^*(i,j)$ (which is $e_i$), the principle of optimality is preserved---i.e., the maximum operator in Equation (\ref{eq:rec}) is guaranteed to ensure optimality of $T(i,j)$.
  
\begin{figure*}
  \centering
  \includegraphics[scale=0.55]{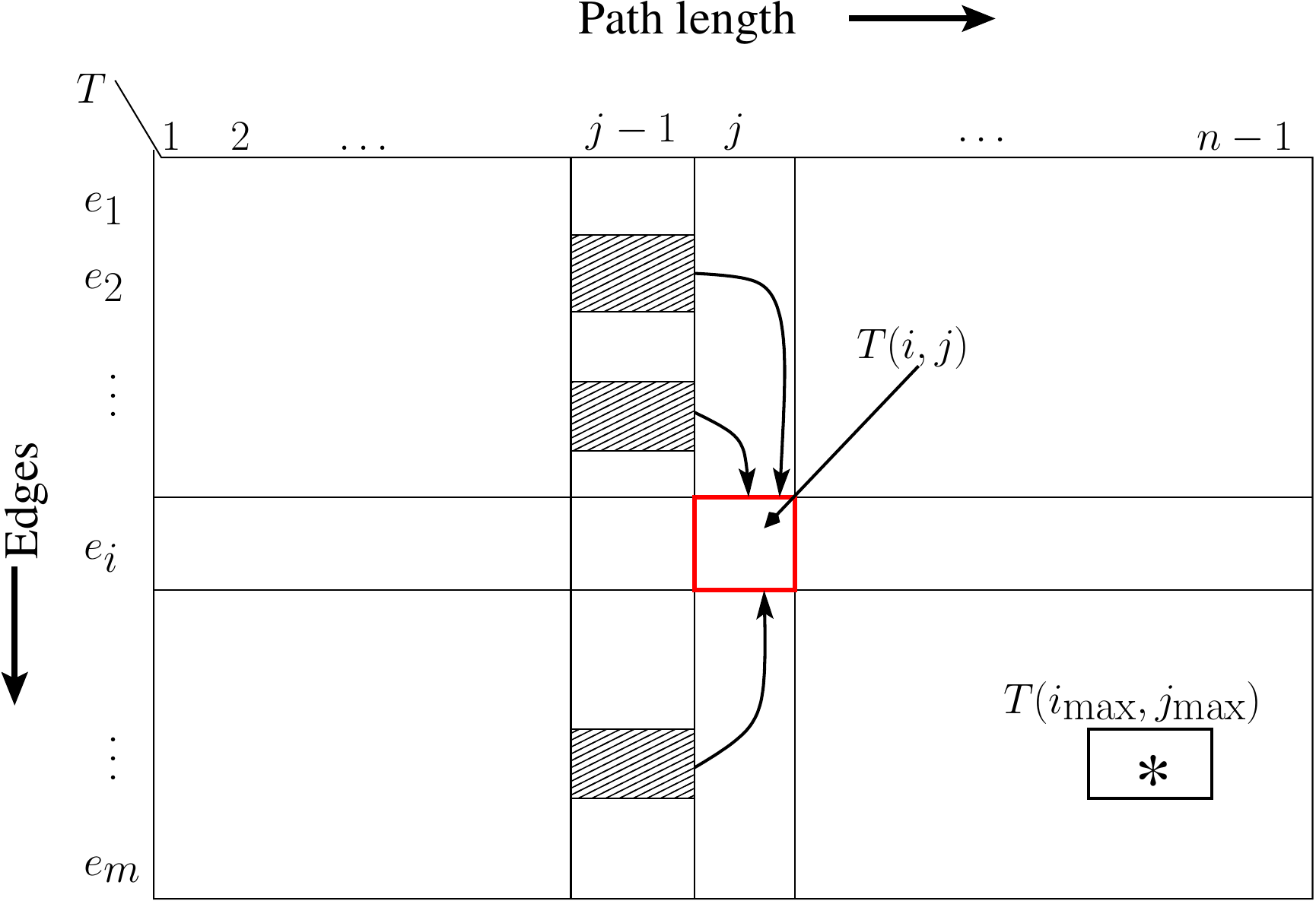}
  \caption{Table $T(i,j)$ for the \maxipsp algorithm.} 
  \label{fig:maxipalgoTbl}
\end{figure*}


\paragraph{Complexity analysis:}
The above dynamic programming algorithm can be implemented to run in $O(mn)$ space and a worst-case time complexity of $O(mnd_{\rm in})$, where $d_{\rm in}$ denotes the maximum indegree of any vertex in $V$. 
\end{proof}

\subsubsection{Algorithmic improvements} \label{sssec:algoimpr}
The above dynamic programming algorithm for \maxipsp for DAGs can be implemented to run in space and time smaller in practice than the worst case limits suggested above. 
First, we note that computing the full table $T$ is likely to be wasteful, as it is likely to be sparse in practice.
The sparsity of $T$ follows from the observation that an interesting path of length $j$ ending at edge $e_i$ can exist only if there exists at least one other interesting path of length $j-1$ ending at one of $e_i$'s predecessor edges. 
We can exploit this property by designing an iterative implementation as follows.

Instead of storing the entire table $T$, we store only the rows (edges), and introduce columns on a ``need basis'' by maintaining a dynamic list $L(e_i)$ of column indices for each edge $e_i$.
\begin{compactenum}[S1)]
  \item
    Initially, we assign $L(e_i)=\{1\}$, as each edge is guaranteed to be in an interesting path of length at least $1$ (the path consisting of the edge by itself).  
  \item \label{stepupdate}
    In general, the algorithm performs multiple iterations; 
    within each iteration, we visit and update the dynamic lists for all edges in $E$ as follows. 
    For every edge $e_{i'}\in \Pred(e_i)$, $L(e_i)=L(e_i) \cup \{ \ell + 1 \, | \, \ell\in L(e_{i'})\}$.
    The algorithm iterates until there is no further change in the lists for any of the edges. 
\end{compactenum}

\smallskip
\noindent The number of iterations in the above implementation can be bounded by the length of the longest path in the DAG (i.e., the diameter $\delta_{\max}$), which is less than $n$. 
Also, we implement the list update from predecessors to successors such that each edge is visited only a constant number of times (despite the varying products of in- and out-degrees at different vertices). 
To this end, we implement the update in S\ref{stepupdate} as a two-step process:
first, performing a union of all lists from the predecessor edges of the form $(*,v)$ so that the merged lists can be used to update the lists of all the successor edges of the form $(v,*)$. 
Thus the work in each iteration is bounded by $O(m)$. 

Taken together, even in the worst-case scenario of $(\delta_{\max}+1)$ iterations, the overall time to construct these dynamic lists is $O(m\delta_{\max})$. 
Furthermore, during the list construction process, if one were to carefully store the predecessor locations using pointers, then the computation of  the $T(i,j)$ recurrence in each cell can be executed in time proportional to the number of \emph{non-empty} predecessor values in the table. 
Overall, this revised algorithm can be implemented to run in time $O(m \delta_{\max} d_{\rm in})$, 
and in space proportional to the number of non-zero values in the matrix. 

Further, the above implementation is also inherently parallel since the list value at an edge in the current iteration depends only on the list values of its predecessors from the previous iteration.

\section{The \kipsp Problem} \label{sec:kIP}

The goal of \kipsp for $k \leq n-1$ is to find a set $\cP$ of edge-disjoint interesting $k$-paths such that the sum of their interestingness scores is maximized.
We show that \kipsp on directed graphs can be solved in polynomial time for $k \leq 2$. 
On the other hand, we show that \kipsp on a directed acyclic graph is NP-Complete for $k \geq 3$ (see Theorem \ref{thm:kge3IPDNPC}).


\medskip
\noindent The smallest value of $k$ for which \kipsp is nontrivial is $2$.
We can solve $2$-IP as a weighted matching problem.
\begin{lemma} \label{lem:12-IP}
  \kipdsp on directed graph $G=(V,E)$ is in P for $k \leq 2$.
\end{lemma}

\begin{proof}
  The case of $k=1$ turns out to be trivial.
  An optimal solution for $1$-IP is obtained by taking $\cP$ as a collection of $m$ interesting $1$-paths each comprised of a single edge.
  These $1$-paths are edge-disjoint by definition, and the need to compare signatures within a path does not arise.
  Since all edge weights $\omega \geq 0$, the total interestingness score $\IScr(\cP)$ is guaranteed to be maximum.
  This optimal solution is unique when $\omega(e) > 0$ for all edges $e \in E$.

  We model the $2$-IP problem ($k=2$) as an equivalent weighted matching problem on an undirected graph $G'=(V',E')$, which we construct as follows.
  We include a vertex $i \in V'$ for each edge $e_i \in E$ in the input graph.
  Hence $|V'| = m$.
  Whenever edges $\{e_i,e_j\} \in E$ form an interesting $2$-path $P_{ij}$ in $G$, we add the undirected edge $(i,j) \in E'$ with its weight $\omega_{ij} = \IScr(P_{ij})$ computed using Equation (\ref{eq:IScr}).
  If both interesting paths $P_{ij} = [e_i,e_j]$ and $P_{ji} = [e_j,e_i]$ are possibly formed by a pair of edges $\{e_i,e_j\} \in E$, we set $\omega_{ij} = \max\{\IScr(P_{ij}),\IScr(P_{ji})\}$.
  Notice that $\omega_{ij} \geq 0$ for all edges $(i,j) \in E'$, and $|E'| \leq m(m-1)/2$.
  A matching $M' \subseteq E'$ in $G'$ corresponds to a set $\cP$ of edge-disjoint interesting $2$-paths in $G$---a vertex $i \in V'$ will be matched with \emph{at most} one other vertex $j \in V'$, and such a match of vertices in $V'$ corresponds to the interesting path $P_{ij} \in \cP$ (or $P_{ji} \in \cP$, but not both).
  It follows that a maximum matching in $G'$ corresponds to an optimal solution to $2$-IP on the input graph $G$.

  The maximum weighted matching problem on an undirected graph with $n$ vertices and $m$ edges can be solved in strongly polynomial time---e.g., Gabow's implementation \cite{Ga1990} of Edmonds' algorithm \cite{Ed1965} runs in $O(nm + n^2 \log n)$ time.
  As such, we can solve $2$-IP by solving the weighted matching problem on the associated graph $G'$ in $O(m^3)$ time.
  Hence $k$-IP is in P for $k \leq 2$.
\end{proof}




We now consider \kipsp for $k \geq 3$ on directed acyclic graphs.
To characterize its complexity, we study the decision version of \kipsp termed \kipd, in which we are given a directed acyclic graph $G=(V,E)$ with edge weights $\omega(e) \geq 0$ and signatures $\Sig(e)$ for $e \in E$ and a target score $s_0 \geq 0$.
The goal is to determine if there exists a collection $\cP$ of edge-disjoint interesting $k$-paths in $G$ whose total interestingness score $\IScr(\cP) = s_0$.

\begin{theorem} \label{thm:kge3IPDNPC}
  \kipdsp on a directed acyclic graph $G=(V,E)$ is NP-complete for $k \geq 3$.
\end{theorem}
\begin{proof}
  Given a collection $\cP$ of interesting $k$-paths in a directed acyclic graph $G=(V,E)$, we can verify they are edge-disjoint, each path has $k$ edges, and signatures are identical for all edges in each path, all in polynomial time.
  We can compute the interestingness score of each $k$-path $P \in \cP$ using Equation (\ref{eq:IScr}) also in polynomial time, and add the $\IScr(P)$ values to compare with $s_0$ to check for equality.
  Hence \kipdsp is in NP.

  We now reduce the \emph{exact $3$-cover} problem (X$3$C) to \tipd.
  We then show a similar reduction for $k \geq 4$ as well, proving \kipdsp is NP-complete for $k \geq 3$.
  The latter case for general $k$ subsumes the case for $k=3$.
  We still present the details for $k=3$ separately, as this case reveals the structure of the general reduction in an arguably simpler setting.
  X$3$C is a version of one of the 21 NP-complete problems originally introduced by Karp \cite{Ka1972}, and is defined as follows.
  Given a set $X$ with $|X|=3q$ elements and a collection $\cC$ of $3$-element subsets of $X$ with $|\cC| = p$, determine if there exists a subset $\cC' \subseteq \cC$ such that each element of $X$ belongs to exactly one member of $\cC'$.
  Notice that such an exact cover $\cC'$ must necessarily have exactly $q$ members. 
  Also, we assume $p \geq q \geq 3$ (else the instance will be trivial).

  Given an instance of X$3$C, we create a directed acyclic graph $G=(V,E)$ for an instance of \tipdsp as follows.
  Each element $x \in X$ corresponds to a unique directed edge in $G$.
  Corresponding to each $3$-element set $\{x,y,z\} \in \cC$, we add to $G$ a directed acyclic graph object as shown in Figure \ref{fig:3_IPD}.
  The edges corresponding to all $x \in X$ are assigned the large weight $\omega = p$ making them ``heavy'' edges, while the rest of the edges are all assigned unit weights.
  Further, we assume $\Sig(e)$ is identical for all edges $e\in E$.
  The three ``V''-shaped $3$-paths in the top of Figure \ref{fig:3_IPD} are referred to as the $x$-, $y$-, and $z$-paths.
  Notice that by this construction, $G$ can have at most $|V| \leq 13p$ vertices and $|E| \leq 12p$ edges.
  
  \begin{figure}[ht!]
    \centering
    \includegraphics[keepaspectratio=yes, width=.7\textwidth]{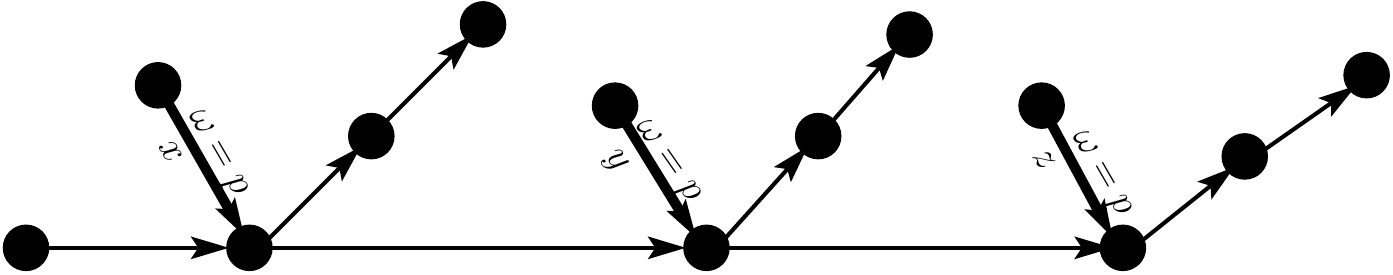}
    \caption{Graph object corresponding to set $\{x,y,z\} \in \cC$ in X$3$C. Thin edges all have $\omega=1$.}
    \label{fig:3_IPD}
  \end{figure}

  Let $\Win = 3(p \log 2 + \log 3 + \log 4) + (\log 2 + \log 3 + \log 4) = 4 \log 24 + 3(p-1) \log 2$, and $\Wout = 3 \log 24$.
  From a graph object as shown above, we observe that $4$ edge-disjoint interesting $3$-paths can be chosen by \tipdsp each with interestingness score $\Win$ if and only if the $x$-, $y$-, and $z$-paths are chosen along with \emph{one} other interesting $3$-path as shown in Figure \ref{fig:3_IPDpathsin}.
  Further, each edge corresponding to an element in $X$ may belong to only one $3$-path.
  Thus, at most $q$ such graph objects may contribute the score of $\Win$ to the total interestingness score.
  The remaining $p-q$ graph objects may contribute a score of at most $\Wout$ corresponding to the selection of the $3$ edge-disjoint interesting $3$-paths shown in Figure \ref{fig:3_IPDpathsout}, which avoid the edges corresponding to any $x \in X$.
  If $q$ such graph objects do contribute $\Win$ each to the total interestingness score, it is clear that the corresponding $q$ triplet elements in $\cC$ form an exact $3$-cover of $X$.
  Further, \tipdsp on $G$ will identify exactly $4q + 3(p-q) = 3p+q$ edge-disjoint interesting $3$-paths with a total interestingness score of exactly $q \Win + (p-q)\Wout  = p \Wout + 3(p-1)q \log 2+q\log 24$.

  \medskip
  \begin{figure}[ht!]
    \begin{minipage}{.55\textwidth}
      \hspace*{0.05in}
      \includegraphics[keepaspectratio=yes, width=0.85\textwidth]{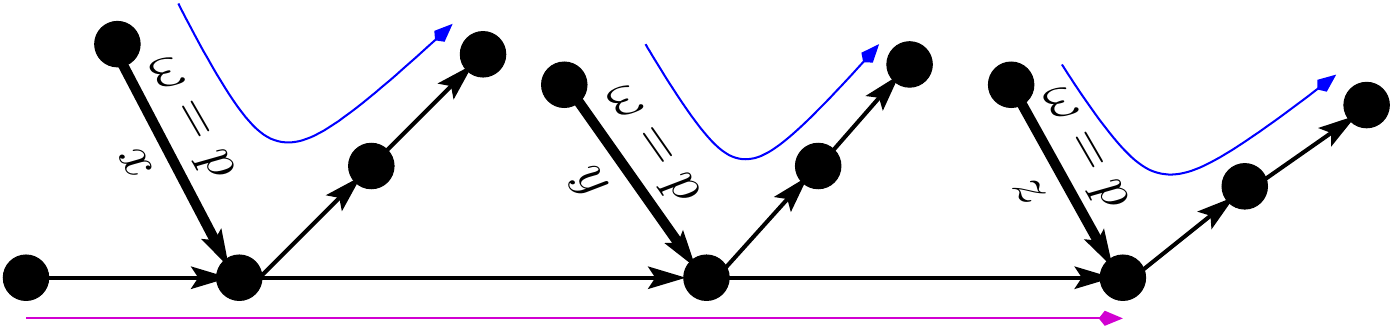}
      \subcaption{$3$-paths for an element in $\cC'$. } \label{fig:3_IPDpathsin}
    \end{minipage}
    \hspace*{-0.72in}
    \begin{minipage}{.55\textwidth}
      \hspace*{0.4in}
     \includegraphics[keepaspectratio=yes, width=0.85\textwidth]{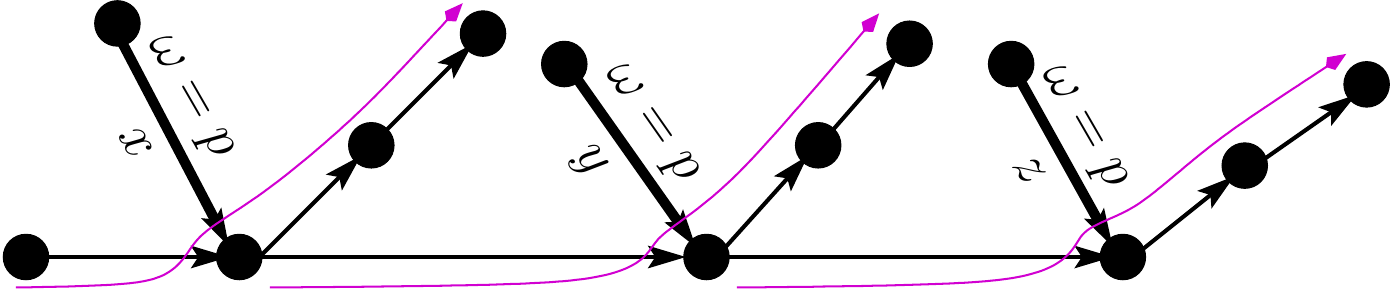}
      \subcaption{$3$-paths for an element not in $\cC'$.} \label{fig:3_IPDpathsout}
    \end{minipage}
    \caption{Choices of interesting $3$-paths in graph objects for elements in $\cC'$ and $\cC \setminus \cC'$.}   \label{fig:3_IPDpaths}
  \end{figure}

  Conversely, if X$3$C has an exact $3$-cover $\cC'$, we choose the $4$ edge-disjoint interesting $3$-paths (recall we assume identical signatures for all edges in $G$) with interestingness score $\Win$ as described above in the graph object in $G$ corresponding to each of the $q$ $3$-element sets in $\cC'$.
  For the $(p-q)$ $3$-element sets in $\cC \setminus \cC'$, we choose the $3$ interesting $3$-paths with interestingness score $\Wout$ each in the corresponding graph objects in $G$.
  This collection of $p$ edge-disjoint interesting $3$-paths in $G$ will have a total interestingness score of exactly $p \Wout + 3(p-1)q \log 2+q\log 24$.

  Thus $X$ has an exact $3$-cover if and only if $\tipdsp$ on $G$ has a target total interestingness score of $s_0 = p \Wout + 3(p-1)q \log 2+q\log 24$, proving $\tipdsp$ is NP-complete.

  \medskip
  We now extend this result to \kipdsp for $k \geq 4$.
  To this end, we reduce the \emph{exact $k$-cover} problem (X$k$C) to \kipdsp for general $k \geq 4$.
  The X$k$C problem is a generalization of X$3$C, and is defined as follows.
  Given a set $X$ with $|X|=kq$ elements and a collection $\cC$ of $k$-element subsets of $X$ with $|\cC| = p$, determine if there exists a subset $\cC' \subseteq \cC$ such that each element of $X$ belongs to exactly one member of $\cC'$.
  Notice that such an exact cover $\cC'$ must necessarily have exactly $q$ members. 
  Also, we assume $p \geq q \geq k$ (else the instance will be trivial).

  Given an instance of X$k$C, we create a directed acyclic graph $G$ for an instance of \kipdsp as follows.
  Each element $x \in X$ corresponds to a unique directed edge in $G$.
  For each $k$-element set $\{x_1,x_2,\dots,x_k\} \in \cC$, we add to $G$ a corresponding directed acyclic graph object as shown in Figure \ref{fig:k_IPD}.
  The edges corresponding to all $x \in X$ are assigned the large weight $\omega = p$ (giving ``heavy'' edges), while the rest of the edges are all assigned unit weights.
  Further, we assume $\Sig(e)$ is identical for all edges $e\in E$.
  The $k$ ``V''-shaped $k$-paths in the top of Figure \ref{fig:k_IPD} are referred to as the $x_1$-, $x_2$-,$\dots x_k$-paths.
  Notice that by this construction, $G$ can have at most $|V| \leq (k(k+1)+1)p$ vertices and $|E| \leq (k(k+1))p$ edges.
  
  \begin{figure}[ht!]
    \centering
    \includegraphics[keepaspectratio=yes, width=.7\textwidth]{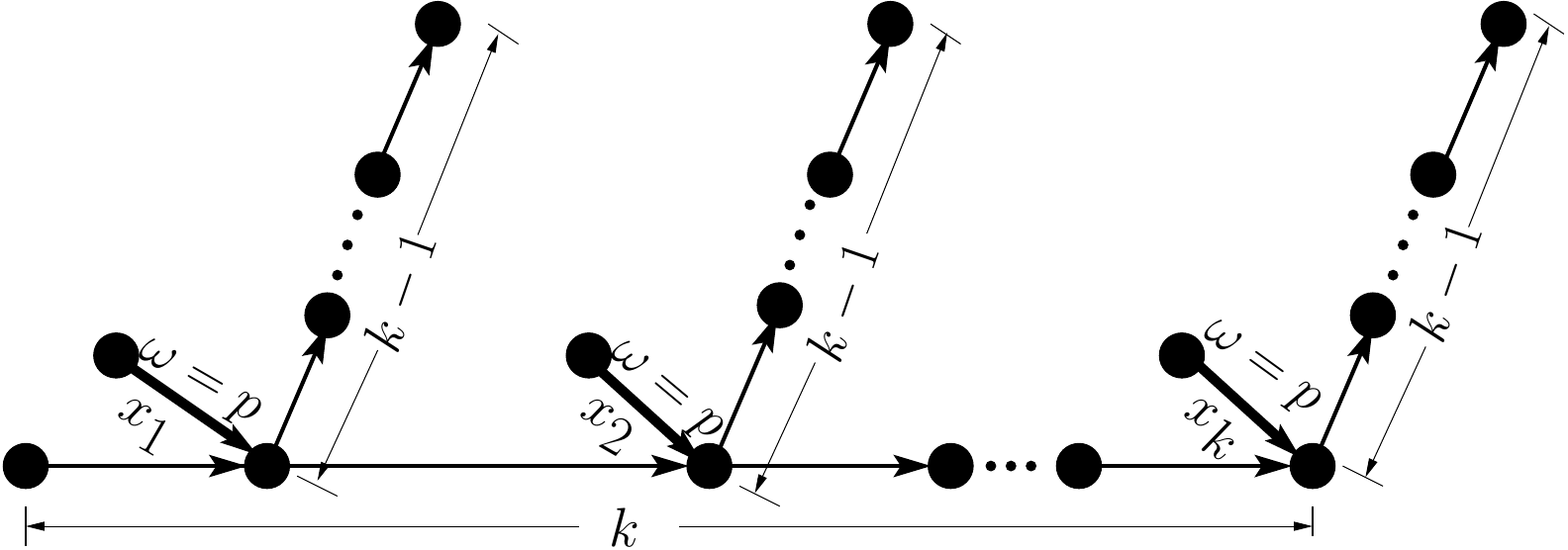}
    \caption{Graph object corresponding to set $\{x_1,x_2,\dots,x_k\} \in \cC$ in X$k$C. Thin edges all have $\omega=1$.}
    \label{fig:k_IPD}
  \end{figure}

  Let $\Win = k(p \log 2 + \log 3 + \log 4+\dots+\log (k+1) ) + (\log 2 + \log 3 + \dots+\log (k+1) ) = (k+1) \log ((k+1)!) + k(p-1) \log 2$, and $\Wout = k \log ((k+1)!)$.
  From a graph object as shown above, we observe that $k+1$ edge-disjoint interesting $k$-paths can be chosen by \kipdsp each with interestingness score $\Win$ if and only if the $x_1$-, $x_2$-, $\dots$ $x_k$-paths are chosen along with \emph{one} other $k$-path as shown in Figure \ref{fig:k_IPD_pathsin}.
  Further, each edge corresponding to an element in $X$ may belong to only one $k$-path.
  Thus, at most $q$ such graph objects may contribute the score of $\Win$ each to the total interestingness score.
  The remaining $p-q$ graph objects may contribute a score of at most $\Wout$ each corresponding to the selection of the $k$ edge-disjoint interesting $k$-paths shown in Figure \ref{fig:k_IPD_pathsout}, which avoid the edges corresponding to any $x \in X$.
  If $q$ such graph objects do contribute $\Win$ each to the total interestingness score, it is clear that the corresponding $q$ $k$-tuple elements in $\cC$ form an exact $k$-cover of $X$.
  Further, \kipdsp on $G$ will identify exactly $(k+1)q + k(p-q) = kp+q$ edge-disjoint interesting $k$-paths with a total interestingness score of exactly $q \Win + (p-q)\Wout  = p \Wout + k(p-1)q \log 2+q\log ((k+1)!)$.

  \medskip
  \begin{figure}[ht!]
    \begin{minipage}{.55\textwidth}
      \hspace*{0.05in}
      \includegraphics[keepaspectratio=yes, width=0.85\textwidth]{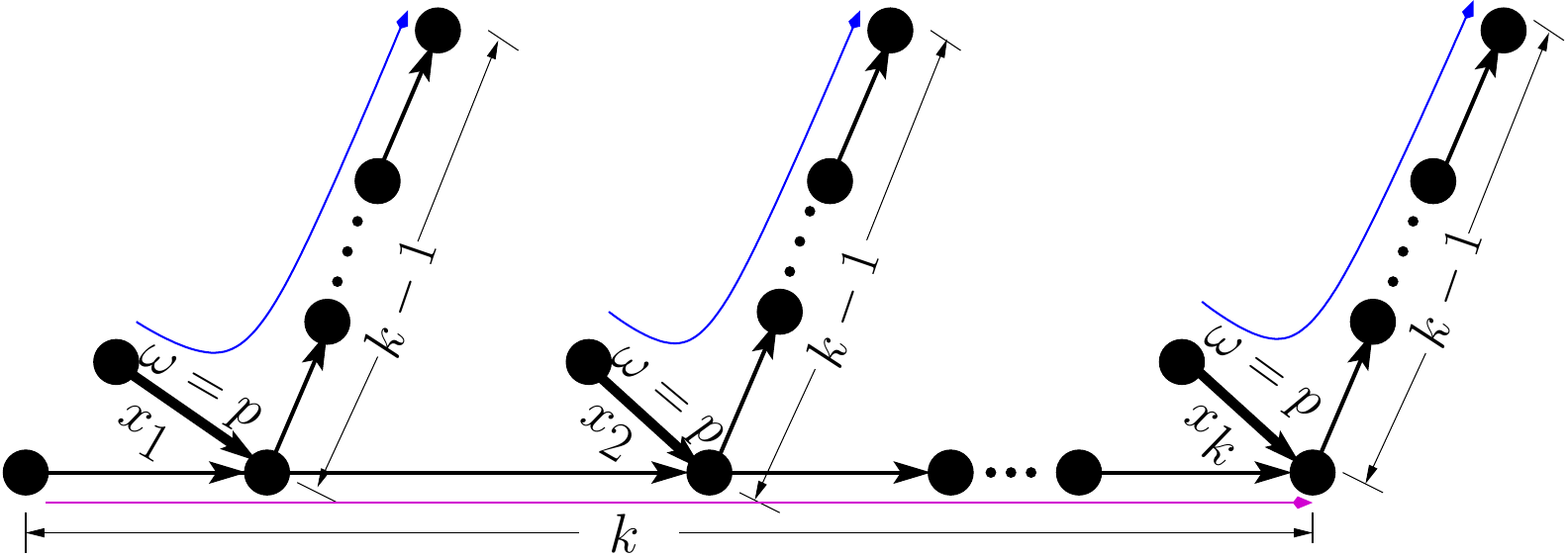}
      \subcaption{$k$-paths for an element in $\cC'$. } \label{fig:k_IPD_pathsin}
    \end{minipage}
    \hspace*{-0.72in}
    \begin{minipage}{.55\textwidth}
      \hspace*{0.42in}
     \includegraphics[keepaspectratio=yes, width=0.85\textwidth]{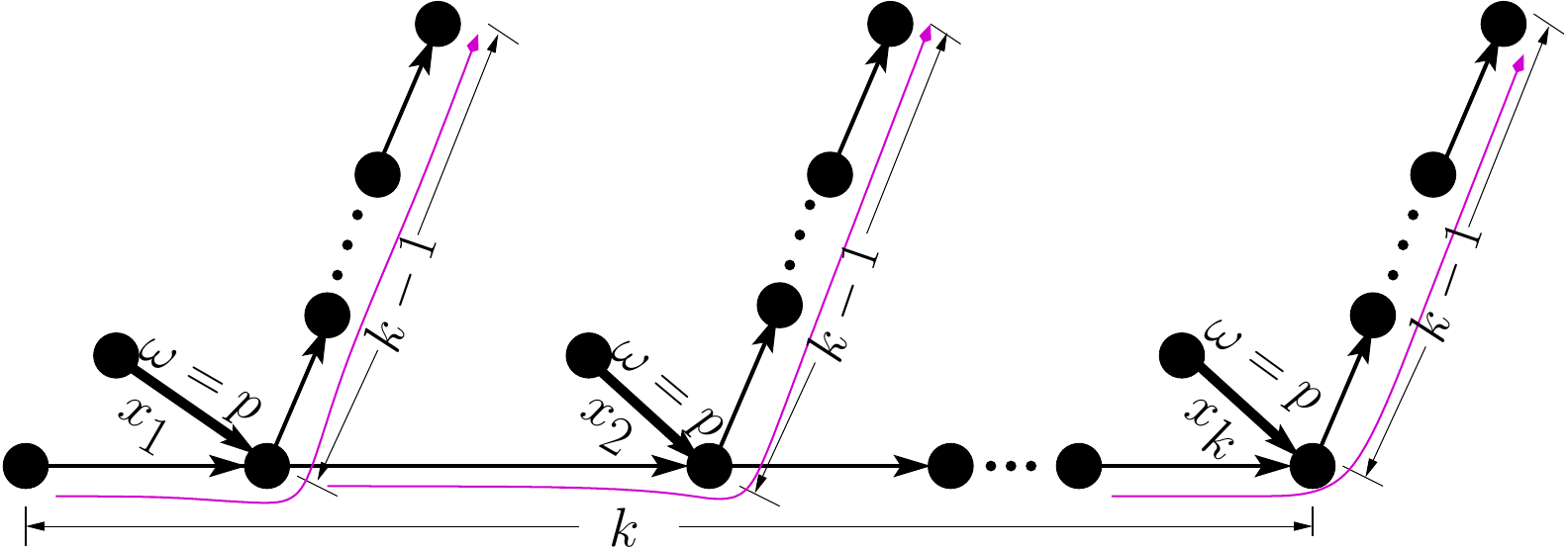}
      \subcaption{$k$-paths for an element not in $\cC'$.} \label{fig:k_IPD_pathsout}
    \end{minipage}
    \caption{Choices of interesting $k$-paths in graph objects for elements in $\cC'$ and $\cC \setminus \cC'$.}   \label{fig:k_IPD_paths}
  \end{figure}

  Conversely, if X$k$C has an exact $k$-cover $\cC'$, we choose the $k+1$ edge-disjoint interesting $k$-paths (again, we assume identical signatures for all edges in $G$) with interestingness score $\Win$ as described above in the graph object in $G$ corresponding to each of the $q$ $k$-element sets in $\cC'$.
  For the $(p-q)$ $k$-element sets in $\cC \setminus \cC'$, we choose the $k$ edge-disjoint interesting $k$-paths with interestingness score $\Wout$ each in the corresponding graph objects in $G$.
  This collection of $p$ edge-disjoint interesting $k$-paths in $G$ will have a total interestingness score of exactly $p \Wout + k(p-1)q \log 2+q\log ((k+1)!)$.

  Thus $X$ has an exact $k$-cover if and only if $\kipdsp$ on $G$ has a target total interestingness score of $s_0 = p \Wout + k(p-1)q \log 2+q\log ((k+1)!)$, proving $\kipdsp$ is NP-complete for $k \geq 4$.
\end{proof}

\noindent Note that since \kipsp is NP-complete for DAGs, it is also NP-complete for directed graphs.

\section{The Interesting Paths (\ip) Problem} \label{sec:IP}
The goal of \ipsp is to find a set $\cP$ of edge-disjoint interesting paths of possibly varying lengths ($1$ to $n-1$) such that the sum of their interestingness scores is maximized.
The determination of the precise complexity class for \ipsp is an open problem. 
However, based on the hardness results for \kipsp (Section~\ref{sec:kIP}), we conjecture that the \ipsp is also intractable.
Here we present an efficient heuristic for \ipsp on DAGs by employing the exact algorithm for \maxipsp on DAGs (in Section \ref{ssec:MaxIPDAG}) as a subroutine.
We also estimate lower and upper bounds on the maximum total interestingness score of \ipsp (see Section \ref{ssec:IP_Bounds}).

\subsection{An efficient heuristic for \ipsp on DAGs} \label{sec:IPheuristic}

We present a polynomial time heuristic to find a set of edge-disjoint interesting paths $\cP$ in a DAG with high total interestingness score.
We do not provide any guarantee on the optimality or (approximation) quality of the collection of interesting paths $\cP$. 

Our method, termed Algorithm~\ref{alg:GreedyIP}, uses a greedy strategy by iteratively calling the exact algorithm for \maxipsp (Section~\ref{ssec:MaxIPDAG}).
The idea is to iteratively detect a maximum interesting path, add it to the working set of solutions, remove all the edges in that path, and recompute \maxipsp on the remaining graph, until there are no more edges left.

\begin{center}
  \begin{minipage}{\linewidth}
    \begin{algorithm}[H]
      \KwIn{DAG $G=(V,E)$ with $\omega(e),\,\Sig(e) ~\forall e \in E$
      }
      \KwOut{A set of edge-disjoint interesting paths $\cP$ in $G$}
      $\mathcal{P} = \emptyset$\\
      \Repeat{$E=\emptyset$}
             {
               $P\gets $ Compute \maxipsp on $G=(V,E)$ and return a most interesting path\\
               $\mathcal{P}\gets \mathcal{P}\cup \{P\}$\\
               Remove edges in $P$ from $E$\\
             }
      \Return $\mathcal{P}$
      \caption{Greedy Heuristic for \ipsp on DAGs}
      \label{alg:GreedyIP}
    \end{algorithm}
  \end{minipage}
\end{center}

\paragraph*{Complexity Analysis:}
The runtime to compute \maxipsp on $G=(V,E)$ in the first step is $O(m\delta_{\rm max}d_{\rm in})$, as described in Section~\ref{sssec:algoimpr}.
Recall that $\delta_{\rm max}$ is the diameter of the DAG and $d_{\rm in}$ is the maximum indegree of any vertex in $V$.
Therefore, if we denote $p$ to be the number of iterations (equivalently, the number of interesting paths found), then the overall runtime complexity is $O(pm\delta_{\rm max}d_{\rm in})$.
However, we expect the performance of the algorithm in practice to be much faster.
Note that at least one edge is, and at most $m$ edges are, eliminated in each iteration, thereby implying $1 \leq p \leq m$ here. 

Consider the worst case of elimination where one edge is eliminated in each iteration, i.e., $p=m$.
The graph must be very sparse in this case, i.e., $m=\Theta(n)$, causing our algorithm for \maxipsp to perform only $O(m)$ work per iteration.
Therefore the overall runtime is $O(m^2)$, or equivalently, $O(n^2)$.

On the other hand, consider the case where the number of edges reduces by a constant factor $c$ at every iteration. 
This setting implies $p=O(\log_c(m))$, while the work performed from one iteration to the next will also continually reduce by a factor of $c$. 
Hence the overall runtime can still be bounded by $O(m\delta_{\rm max}d_{\rm in})$, the cost of \maxip. 
Further, from an application standpoint, such a greedy iterative approach can be terminated whenever an adequate number of ``top'' interesting paths are identified. 

\subsubsection{An efficient heuristic for \kipsp on DAGs} \label{sssec:kIPheuristic}
The above heuristic for \ipsp can be easily modified to devise a heuristic for \kipsp on DAGs. 
Algorithm~\ref{alg:GreedykIP} summarizes the main steps. 
The main idea is to modify the exact algorithm for \maxipsp on a DAG such that it initializes a recurrence table of size $m\times k$ (instead of $m\times (n-1)$), and then use that table to iteratively compute \maxipsp paths.
The only constraint here is that each such \maxipsp path should originate from the $k$-th column during backtracking, so that paths output are guaranteed to be of length $k$. 
The runtime is bounded by $O(mrd_{\rm in})$, where $r=\min\{k,\delta_{\rm max}\}$.

\begin{center}
\begin{minipage}{\linewidth}
  \begin{algorithm}[H]
    \KwIn{DAG $G=(V,E)$ with $\omega(e),\,\Sig(e) ~\forall e \in E$
    }
    \KwOut{A set of edge-disjoint interesting $k$-paths $\cP$ in $G$}
    Initialize an $m\times k$ table $T^\prime$\\
    $\mathcal{P} = \emptyset$\\
    \Repeat{$E=\emptyset$}
           {
             $P^\prime\gets $ Compute \maxipsp on $G=(V,E)$ using $T^\prime$, and return\\
             \hspace*{0.35in} a most interesting path \emph{ending in column $k$}\\
             $\mathcal{P}\gets \mathcal{P}\cup \{P^\prime\}$\\
	     Remove edges in $P^\prime$ from $E$\\
           }
    \Return $\mathcal{P}$
    \caption{Greedy Heuristic for \kipsp on DAGs}
    \label{alg:GreedykIP}
  \end{algorithm}
\end{minipage}
\end{center}

\begin{remark}
  Ideas used in Algorithms \ref{alg:GreedyIP} and \ref{alg:GreedykIP} could be combined to develop a heuristic for \textsc{AtLeast}-\kip, a modified version of \kipsp that seeks to find a collection of interesting paths in $G$ where each path has \emph{at least} $k$ edges.
  We have implemented this heuristic in our software suite (see Section \ref{sec:Implementation} for details).
\end{remark}

\subsection{Bounds for \ip}
\label{ssec:IP_Bounds}
Let $\cP^*$ represent an optimal set of paths for an instance of \ip.
We derive upper and lower bounds on its total interestingness score $\IScr(\cP^*)$.

Let $P_{\max}(i)$ denote a maximum interesting path (of arbitrary length) ending at a given edge $e_i\in E$. 

\begin{lemma}
 $\IScr(\cP^*)\leq \sum_{e_i\in E}{\IScr(P_{\max}(i))}$.
\end{lemma}
\begin{proof}
Consider an arbitrary path $P\in\cP^*$.
We first note that individual paths that are members of an optimal solution ($\cP^*$) for the \ipsp problem can end at any arbitrary \emph{non-source} vertex in $G=(V,E)$ (see Figure \ref{fig:ip_boundary} for an illustration). 

\medskip
\begin{figure}[!ht]
\centering
\includegraphics[scale=0.6]{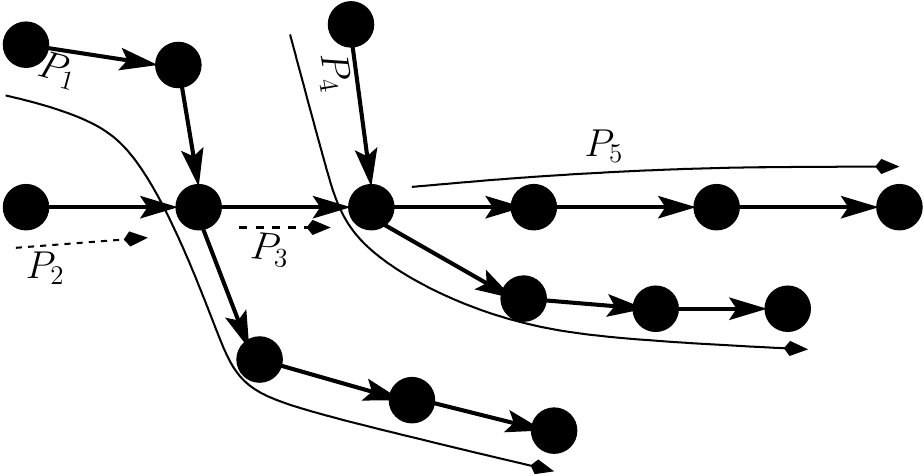}
\caption{Five interesting paths labeled $P_1,\dots,P_5$.
The paths $ P_2$ and $P_3$ end at non-source vertices.}
\label{fig:ip_boundary}
\end{figure}

Without loss of generality, let us assume that the input graph contains only vertices with degree at least one (as vertices with degree zero cannot contribute to any interesting path).
We consider two sub-cases:

\medskip
\emph{Case A: No two maximum scoring paths ending at two different edges $e_i$ and $e_j$ intersect, i.e., $P_{\max}(i)\cap P_{\max}(j)=\phi$, $\forall i,j\in [1,m]$ and $i\ne j$}.\\
This case can occur only if the number of edges ($m$) is equal to the number of source vertices. 
This setting implies that $\cP^*$ is comprised of $m$ paths, where each path $P \in \cP^*$ is a unique edge $e_i\in E$. 
Therefore, $\IScr(\cP^*)=\sum_{e_i\in E}{\IScr(P_{\max}(i))}$ in this case.

\emph{Case B: There exists at least two maximum interesting paths ending at two different edges $e_i$ and $e_j$ that intersect, i.e., $P_{\max}(i)\cap P_{\max}(j)\ne\phi$}.\\
This case implies that at least one of these two paths is \emph{not} a member of $\cP^*$ (by definition of the \ipsp problem);
let this non-member path be $P_{\max}(j)$ without loss of generality. 
Since all edges are covered by $\cP^*$ by definition of \ip, there still has to exist an alternative path ending in $e_j$ that is either directly contained in $\cP^*$ or contained as a \emph{subpath} of a longer path in $\cP^*$;
let us refer to this alternative path as $P'(j)$.  
Since $P_{\max}(j)$ is an optimal interesting path ending at edge $e_j$, $\IScr(P'(j))\leq \IScr(P_{\max}(j))$. 
In other words, the contribution of $\cP'(j)$ to $\IScr(\cP^*)$ cannot exceed the contribution of $P_{\max}(j)$ to $\IScr(\cP^*)$.
Therefore, $\IScr(\cP^*)\leq \sum_{e_i\in E}{\IScr(P_{\max}(i))}$ in this case as well.
\end{proof}

We now present a lower bound for $\IScr(\cP^*)$.

\begin{lemma}
$\IScr(\cP^*)\geq \sum_{e_i\in E}(\omega(e_i) \times \log 2)$.
\end{lemma}
\begin{proof}
Since $\cP^*$ covers all edges in $E$, a trivial (albeit not necessarily optimal) solution $\cP'$ for \ipsp can be constructed by including every edge as a distinct interesting path in the graph, i.e., $\cP'=\{{e_i} | 1 \leq i \leq m \}$.
Therefore, $\IScr(\cP^*)\geq \IScr(\cP')=\sum_{e_i\in E}(\omega(e_i) \times \log 2)$ follows from Equation (\ref{eq:IScr}).
\end{proof}

\section{Implementation and Testing}
\label{sec:Implementation}

We have implemented the core algorithms  for long interesting path detection as part of the HYPPO-X repository, which includes our open source implementation of the Mapper framework.
The HYPPO-X software repository is open source, and is available at \href{https://xperthut.github.io/HYPPO-X}{https://xperthut.github.io/HYPPO-X}.
The core computational modules are implemented in C++ and the visualization modules are implemented using the Javascript visualization library, D3~\cite{teller2013data}. 

\smallskip
Currently, the HYPPO-X repository includes implementations for the following algorithms presented in this paper:
\begin{compactenum}\itemsep=-0.05ex
\item
  the optimized exact algorithm for \maxipsp on DAGs described in Section~\ref{sssec:algoimpr};
\item 
  the greedy iterative heuristic for \ipsp on DAGs described in Algorithm~\ref{alg:GreedyIP}; and
\item 
  greedy iterative heuristics for \kipsp and \textsc{AtLeast}-\kip on DAGs as described in Section~\ref{sssec:kIPheuristic}.
\end{compactenum}

\subsection{Experimental testing}
\label{sec:Experimental}
We have been testing and evaluating our implementations on data sets obtained from an ongoing plant phenomics research project. 
Phenomics is an emerging branch of modern biology that involves the analysis of multiple types of data (genomic, phenotypic, and environmental) acquired using high-throughput technologies. 
A core goal of plant phenomics research is to understand how different crop genotypes (G) interact with their environments (E) to produce varying performance traits (phenotypes (P));
this goal is often summarized as $G\times E\rightarrow P$ \cite{martin2013plant}.

We present sample results from our analysis of a real world maize data set (accessible from the HYPPO-X repository).
This data set consists of phenotypic and environmental measurements for two maize genotypes (abbreviated here for simplicity as A and B), grown in two geographic locations (Nebraska (NE) and Kansas (KS)).
The data consists of daily measurements of growth rate alongside multiple environmental variables, over the course of the entire growing season (100 days). 
In our analysis, each ``point'' refers to a unique [genotype, location, time] combination.
Consequently, the above data set consists of 400 points. Here, ``time'' was measured as Days After Planting (DAP). 

Figure~\ref{fig:phenomics-result} shows an example screenshot of an output generated by running our greedy heuristic for \textsc{AtLeast}-\kip on a subset of the above data set containing only genotype B points.
For the purpose of this analysis, we used humidity and DAP as two filter functions (i.e., independent variables) and growth rate as our target function (i.e., dependent variable) for crop performance.
\begin{figure}[htb!]
\centering
\includegraphics[keepaspectratio=yes, width=\textwidth]{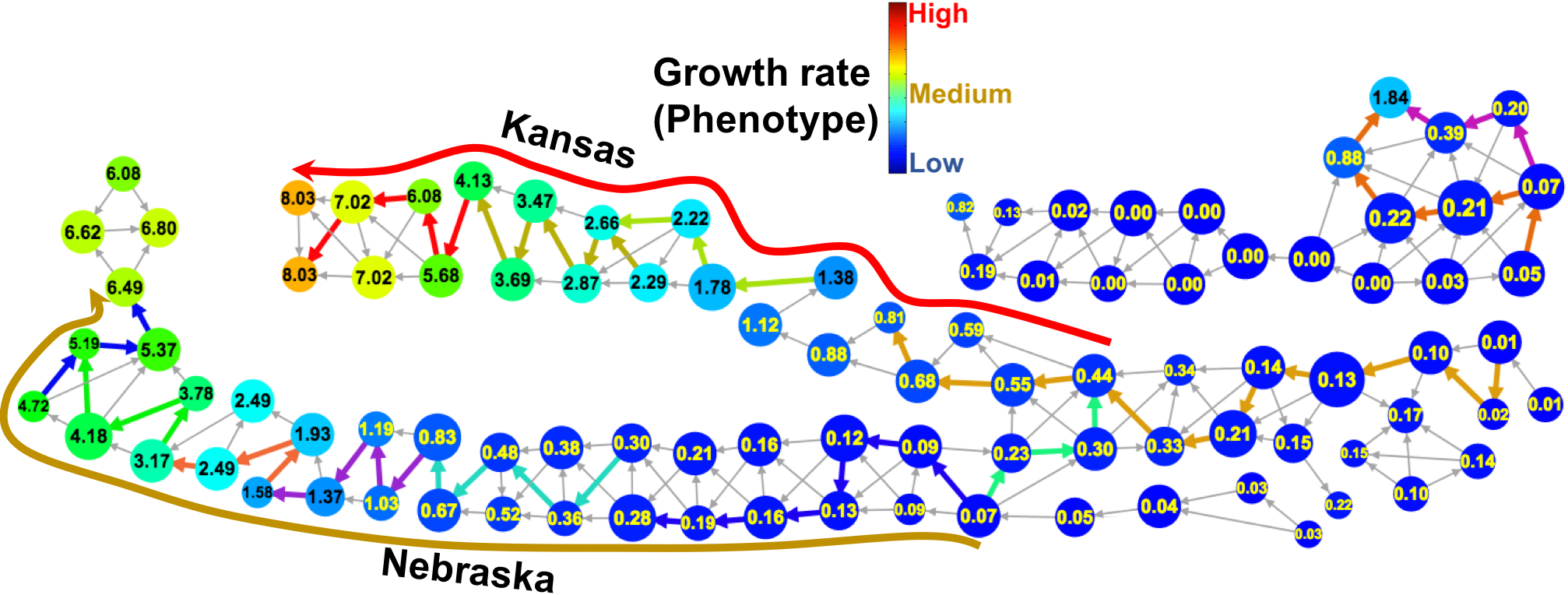}
\caption{
A screenshot of our HYPPO-X tool for identifying interesting paths in Mapper. 
The results correspond to our experiments with a plant phenomics data set (real world maize crop
data).
An interactive version of this visualization can be accessed at \href{https://xperthut.github.io/HYPPO-X}{https://xperthut.github.io/HYPPO-X}.
}
\label{fig:phenomics-result}
\end{figure}

Our tool uncovered two long interesting paths---the red-colored path corresponds to a trail of clusters
from the Kansas subpopulation, whereas the orange-colored path originates from the Nebraska subpopulation.
Note that the location information was \emph{not} used {\it a priori} by our framework---our path detection was entirely unsupervised! 
More importantly, the two long interesting paths helped us detect growth rate variations that are potentially tied to specific variations in both humidity and DAP (time), between the crops of these two locations (KS vs NE).
We are currently working with our plant science collaborators to better understand the basis for these observed variations, with an aim to postulate hypotheses. 

Visualizations of the outputs generated in the above analysis are available on the HYPPO-X repository.
More details about this plant phenomics use case can be found in our related application-oriented paper~\cite{KaKaKrSc2017}.

\section{Discussion} \label{sec:discussion}

We have proposed a general framework for quantifying the significance of features in the Mapper in terms of their interestingness scores.
The associated optimization problems \maxip, \kip, and \ipsp constitute a new class of nonlinear longest path problems on directed graphs.
We have not characterized the complexity of problem \ip.
Judging from the fact that \kipsp is NP-complete even on DAGs, we suspect \ipsp is NP-complete as well.

Our framework for quantifying interesting paths could be modified to quantify branches in flares as well as holes.
When the graph $G$ is a DAG, two interesting $k$-paths that start and end at the same pair of vertices could be characterized as a $2k$-hole, i.e., a cycle with $2k$ edges.
One could alternatively use persistent homology tools to characterize holes---by identifying ``long'' generators around them.
A 2-way branching in a flare could be identified by two interesting $k$-paths where one path starts off from a vertex in the middle of the other path.
Alternatively, we could generalize the definition of interestingness score for a path in Equation (\ref{eq:IScr}) to that of a 2-way branch.
Subsequently, we could seek to solve the related optimization problems of identifying the most interesting $2$-way flare, or to identify a collection of $2$-way flares whose total interestingness score is maximized.

While we distinguished the clustering function $g$ from the filter functions $f_i$ for $i=1,\dots,h$, this distinction is not critically used in our framework.
As indicated in Remark \ref{rem:gclust}, one could use $g$ simultaneously as a filter function along with the $f_i$'s, and the overall analysis should still carry through.
More generally, details of how to implement clustering within the Mapper framework has not received much research attention.
In initial work on phenomics \cite{KaKaKrSc2017}, we obtained better results when using $g$ alone to cluster within Mapper (rather than clustering using several, or even all, of the variables).
It would be interesting to characterize the stability of Mapper to varying settings of clustering employed in its construction.
For instance, could we identify a ``small'' subset of variables for use in clustering within Mapper that is optimal in a suitable sense?

While we have proposed an efficient heuristic for \ipsp on DAGs, we are not able to certify the quality of solution obtained by this method.
On the other hand, could we devise approximation algorithms for \ipsp or \kip?
One might have to work under some simplifying assumptions on the distribution of weights $\omega(e)$ or on the structure of the graph $G$.
The simplest case to consider appears to that of \ipsp on a DAG with unit weights on all edges.

We study interestingness of features in a given single Mapper.
A natural extension to consider would be to characterize the stability of the highly interesting features.
Could we incorporate our interestingness scores into the mathematical machinery recently developed to obtain results on stability and statistical convergence of the 1-D Mapper \cite{CaMiOu2017,CaOu2017}?
Another natural extension to explore would be to define and efficiently identify interesting \emph{surfaces} (or manifolds in general) in higher dimensional mappers.

\paragraph{Acknowledgments:}

This research is supported by the NSF grant DBI-1661348.
Krishnamoorthy thanks Fr\'ed\'eric Meunier for discussion on the complexity of \maxipsp while visiting MSRI.


\bibliographystyle{plainurl}


\end{document}